\documentclass[9pt,shortpaper,twoside,web]{ieeecolor}
\usepackage{generic}

\usepackage{cite}
\usepackage{amsmath,amssymb,amsfonts}
\usepackage{algorithmic}
\usepackage{graphicx}
\usepackage{textcomp}
\usepackage{graphicx}
\usepackage{color}
\definecolor{darkgreen}{rgb}{0,0.4,0}
\usepackage{float}
\definecolor{darkgreen}{rgb}{0,0.6,0}
\usepackage{authblk}
\usepackage[colorlinks,linkcolor=blue,citecolor=darkgreen]{hyperref}
\usepackage[latin1]{inputenc}
\usepackage{amsmath}
\usepackage{graphicx}
\usepackage{amssymb}
\usepackage{verbatim}
\usepackage{epsfig}%%
\usepackage[labelfont=bf]{caption}
\usepackage{enumerate}
\usepackage{subfig}
\usepackage{epstopdf}
\usepackage{relsize,exscale}

\usepackage{booktabs} %% tables with three lines

\newtheorem{definition}{Definition}
\newtheorem{assumption}{Assumption}

\newtheorem{remark}{Remark}

\newtheorem{corollary}{Corollary}
\newtheorem{proposition}{Proposition}

\def\begquo{\begin{quote}}
\def\endquo{\end{quote}}
\def\begequarr{\begin{eqnarray}}
\def\endequarr{\end{eqnarray}}
\def\begequarrs{\begin{eqnarray*}}
\def\endequarrs{\end{eqnarray*}}
\def\begarr{\begin{array}}
\def\endarr{\end{array}}
\def\begequ{\begin{equation}}
\def\endequ{\end{equation}}
\def\lab{\label}
\def\begdes{\begin{description}}
\def\enddes{\end{description}}
\def\begenu{\begin{enumerate}}
\def\begite{\begin{itemize}}
\def\endite{\end{itemize}}
\def\endenu{\end{enumerate}}

\def\lef[{\left[\begin{array}}
\def\rig]{\end{array}\right]}

\def\begcen{\begin{center}}
\def\endcen{\end{center}}
\def\begrem{\begin{remark}\rm}
\def\endrem{\end{remark}}
\def\begdef{\begin{definition}}
\def\enddef{\end{definition}}
\def\begpro{\begin{proposition}}
\def\endpro{\end{proposition}}
\def\begfac{\begin{fact}}
\def\endfac{\end{fact}}
\def\begass{\begin{assumption}}
\def\endass{\end{assumption}}
\def\begsubequ{\begin{subequations}}
\def\endsubequ{\end{subequations}}

\def\begmat#1{\begin{bmatrix}#1\end{bmatrix}}

\def\begalis#1{\begin{align*}{#1}\end{align*}}
\def\caly{{\cal Y}}

\def\calh{{\cal H}}

\def\call{{\cal L}}

\def\hatthe{\hat{\theta}}

\def\dothatthe{\dot{\hat{\theta}}}

\def\liminf{\lim_{t \to \infty}}

\def\L2e{{\cal L}_{2e}}

\def\rea{\mathbb{R}}

\def\adj{\mbox{adj}}

\def\hal{{1 \over 2}}

\def\trace{\mbox{trace}}
\def\TAC{{\it IEEE Trans. Automatic Control}}

\def\AUT{{\it Automatica}}

\usepackage{color}

\usepackage[prependcaption,colorinlistoftodos]{todonotes}

\def\BibTeX{{\rm B\kern-.05em{\sc i\kern-.025em b}\kern-.08em
    T\kern-.1667em\lower.7ex\hbox{E}\kern-.125emX}}
\begin{document}
\title{
Generation of New Exciting Regressors for Consistent On-line Estimation of Unknown Constant Parameters
}

\author{%Alexey Bobtsov, Bowen Yi, Romeo Ortega and Alessandro Astolfi
Alexey Bobtsov, \IEEEmembership{Senior Member, IEEE}, Bowen Yi, Romeo Ortega, \IEEEmembership{Life Fellow, IEEE}

Alessandro Astolfi, \IEEEmembership{Fellow, IEEE}%
% <-this % stops a space
% <-this % stops a space
\thanks{This work has been partially supported by the European Union's Horizon 2020 Research and Innovation Programme under grant agreement No. 739551 (KIOS CoE), the Italian Ministry for Research in the framework of the 2017 Program for Research Projects of National Interest (PRIN), Grant no. 2017YKXYXJ, the Open Research Project of the State Key Laboratory of Industrial Control Technology, Zhejiang University, China (No. ICT2022B68), and the Russian Science Foundation, project no. 18-19-00627, https://rscf.ru/project/18-19-00627. \emph{(Corresponding author: Bowen Yi)}}

\thanks{A. Bobtsov is with Hangzhou Dianzi University (HDU), 310018, Hangzhou, China and the Faculty of Control Systems and Robotics, ITMO University, Kronverkskiy av. 49, St. Petersburg, 197101, Russia ({\tt bobtsov@mail.ru})}
\thanks{B. Yi is with Australian Centre for Field Robotics, The University of Sydney, Sydney, NSW 2006, Australia. ({\tt bowen.yi@sydney.edu.au})}%
\thanks{R. Ortega is with Departamento Acad\'emico de Sistemas Digitales, ITAM, Ciudad de M\'exico, M\'exico ({\tt romeo.ortega@itam.mx})}
\thanks{A. Astolfi is with the Department of Electrical and Electronic Engineering, Imperial College London, London SW7 2AZ, UK and the DICII,
Universita di Roma ``Tor Vergata'', Via del Politecnico 1, 00133 Roma, Italy ({\tt a.astolfi@ic.ac.uk})}
}

\maketitle

\begin{abstract}
The problem of estimating constant parameters from a standard vector linear regression equation in the absence of sufficient excitation in the regressor is addressed. The first step to solve the problem consists in transforming this equation into a set of scalar ones using the well-known dynamic regressor extension and mixing technique. Then a novel procedure to generate new scalar exciting regressors is proposed. The superior performance of a classical gradient estimator using this new regressor, instead of the original one, is illustrated with comprehensive simulations.
\end{abstract}


%%%%%%%%
\section{Introduction}
\label{sec1}
%%%%%%%%
%
It is well known that consistent estimation of constant parameters from a linear regression equation (LRE) with a gradient (or least squares) estimator is possible only if the regressor satisfies certain excitation conditions. A classical result shows that exponential convergence is possible if and only if the regressor verifies the {\em persistence of excitation} (PE) condition \cite{SASBODbook}---which is a uniform observability property for the associated linear time-varying (LTV) system. It has recently been shown in \cite{BARORT,PRA} that asymptotic (but not exponential) convergence is guaranteed under the strictly weaker condition of generalized PE, the definition of which may be found in \cite[Proposition 6]{BARORT}.

Unfortunately, the PE (or the generalized PE) property is rarely satisfied in applications, hence the interest to propose new adaptation algorithms that ensure parameter convergence without PE. This research line has been intensively pursued in the last few years and some recent adaptive schemes, in which the PE assumption is obviated via the incorporation of some {\em off-line} data manipulation, have been reported in \cite{CHOetal,PANYU,CHOSHI,ROYetal}---see also \cite{ORTNIKGER} for a recent survey.

In this paper we are interested in {\em on-line} estimation using recursive algorithms. It is well known that, in contrast with off-line estimation schemes, on-line estimation provides, via the accumulation of past measurements and noise averaging, stronger robustness properties. Moreover, if the adaptation gain of the estimator remains bounded away from zero---a property usually referred as {\em alertness}\footnote{It is well known \cite[Section 2.3.2]{SASBODbook} that, due to the so-called covariance wind-up problem, the alertness property is lost in standard least-squares estimators. Therefore, we concentrate on gradient-descent schemes.}---it has the ability of tracking slowly time-varying parameters. We concentrate our attention to the case of a {\em single} uncertain parameter. Our main motivation to study the scalar case stems from the recent development of the {\em dynamic regressor extension and mixing} (DREM) estimator \cite{ARAetal}, which is a procedure that generates, from a $q$-dimensional LRE, $q$ scalar LREs for each of the unknown parameters. It has been observed in several applications that the absence of excitation stymies the successful use of DREM. For instance, in \cite{araetalijacsp19} it is shown that consistent estimation of the parameters of a linear time-invariant (LTI) system with DREM is possible if and only if the original regressor is PE. Actually, since the key scalar function that defines the convergence properties of the gradient estimator in DREM is the determinant of the extended regressor, this converges in many cases {\em to zero}, hence we only have excitation on a finite interval.

Our main contribution is to propose a procedure to generate, from a scalar LRE, a new scalar LRE in which the {\em new regressor} satisfies the excitation property of non-square integrability, even in the case in which the original regressor is not sufficiently exciting---for instance an exponentially decaying signal. It is shown in \cite{ARAetal} that non-square integrability of the (scalar) regressor is necessary and sufficient to ensure asymptotic convergence, which becomes exponential imposing the PE condition. Instrumental for the construction of the new LREs, that include some free signals, is to borrow the key idea of the {\em parameter estimation based observer} (PEBO) proposed in \cite{ORTetalscl}, later generalized in \cite{ORTetalaut}, which is a constructive procedure to design state observers for state-affine nonlinear systems. Then, applying the energy pumping-and-damping injection principle of \cite{YIetal}, which was proposed as a passivity-based orbital stabilization technique for port-Hamiltonian systems, we select these signals to guarantee the desired excitation properties of the new regressor.

The remainder of the paper is organized as follows. In Section \ref{sec2} we briefly recall the DREM procedure and review the problem of parameter estimation for a scalar LRE. Section \ref{sec3} presents the procedure to generate the new LRE with some free signals, which are selected in Section \ref{sec4} to comply with the excitation injection requirement. Simulation results, which illustrate the {\em superior performance} of the classical gradient estimator using the new regressor, instead of the original one, are presented in Section \ref{sec5}. The paper is wrapped-up with concluding remarks and a discussion on future research in Section \ref{sec6}.\\

\noindent {\bf Notation.} $I_n$ is the $n \times n$ identity matrix. For a vector $x \in \rea^n$ we denote the Euclidean norm as $|x|:=\sqrt {x^\top x}$. $\call_1$ and $\call_2$ denote the absolute integrable and square integrable function spaces, respectively, and $\call_\infty$ represents the vector space of essentially bounded functions.

%
%%%%%%%%
\section{Gradient Estimation of a Single Parameter}
\label{sec2}
%%%%%%%%
%
In this paper we deal with the problem of on-line estimation of the unknown {\em constant} parameters $\theta \in \rea^q$ appearing in an LRE of the form
\begequ
\lab{veclre}
w=\psi^\top\theta,
\endequ
where $w(t) \in \rea$ and $\psi(t) \in \rea^q$ are {\em measurable} signals. This problem appears in several applications including system identification  \cite{LJUbook} and adaptive control \cite{IOASUNbook,NARANNbook,SASBODbook} in which, as discussed in Section \ref{sec1}, a key requirement to achieve their objectives is that the regressor $\psi$ is PE, a condition that is rarely satisfied in practice. Our task is then to generate new LREs that satisfy the PE requirement.
\subsection{Generation of scalar LRE}
\label{subsec21}
%%%%%%%%
%

The first step in our design is to apply the DREM procedure \cite{ARAetal} to obtain  $q$ {\em scalar} LREs, one for each of the unknown parameters. Towards this end, we introduce a {\em linear, single-input $q$-output, bounded-input bounded-output (BIBO)--stable} operator $\calh$ and define the vector $W \in \rea^m$ and the matrix $\Psi \in \rea^{q \times q}$ as
$$
W  := \calh[w],\;\Psi :=\calh[\psi^\top].
$$
Clearly, because of linearity and BIBO stability, these signals satisfy
$$
W = \Psi \theta.
$$
At this point the key step of regressor ``mixing" of the DREM procedure is used to obtain a set of $q$ {\em scalar} equations multiplying from the left the vector equation above by the {\em adjunct matrix}, denoted $\adj\{\cdot\}$, of $\Psi$ to get
\begequ\label{lre}
y_i = \Delta \theta_i,\;i \in \bar q:=\{1,2,\dots,q\},
\endequ
in which we have defined
\begalis{
\Delta & :=\det \{\Psi\},\\
y &:= \adj\{\Psi\} W.
}
This fundamental modification has numerous advantages and DREM has been instrumental to solve many, previously open, problems---see \cite{ORTNIKGER,ORTetaltac,ORTetalaut} for a recent account of some of these results.
\subsection{Properties of the gradient estimator}
\label{subsec22}
%%%%%%%%
%
Motivated by the developments above, in the remaining part of the paper we consider scalar LREs of the form \eqref{lre}. The following property of the gradient estimator is easy to establish \cite{ARAetal}.

\begin{proposition}
\lab{pro1}
Consider the scalar LREs \eqref{lre} with\footnote{To simplify the notation we omit the subindex $(\cdot)_i$.} $y(t) \in \rea$ and $\Delta(t) \in \rea$ bounded, measurable, signals and $\theta \in \rea$ an {\em unknown} parameter estimated on-line via the gradient descent adaptation algorithm
\begequ
\lab{graest}
\dothatthe=-\gamma \Delta(\Delta \hatthe - y),
\endequ
with $\gamma>0$ the adaptation gain.
\begin{enumerate}
\item[{\bf P1}] The following equivalence holds true:
$$
\liminf |\hat \theta(t)-\theta|=0~~\iff~~\Delta \notin \call_2.
$$
\item[{\bf P2}] The convergence of the estimate $\hat \theta$ to $\theta$ is exponential if  and only if $\Delta$ is PE, that is, if there exist $T>0$ and $\delta>0$ such that
$$
\int_t^{t+T}\Delta^2(\tau) d\tau \geq \delta,\;\forall t\geq 0.
$$
\end{enumerate}
\end{proposition}
%
%%%%%%%%
\section{Generation of the New LRE}
\label{sec3}
%%%%%%%%
%

In this section, in order to generate a new LRE, we apply some filters to both sides of the original LRE \eqref{lre} with some free terms to comply with the excitation requirement. Then, in Section \ref{sec4} we study how to design these terms to fulfill this task.

In the sequel we apply the construction used in {generalized parameter estimation based observer} (GPEBO) \cite{ORTetalaut} to create a new LRE from the scalar LRE \eqref{lre}.
\begin{proposition}
\lab{pro2}
Consider the scalar LRE \eqref{lre}. There exists a measurable signal $\caly_2(t) \in \rea$ such that the new LRE
\begequ\label{newlre}
\caly_2  = \Phi_{21} \theta
\endequ
holds, in which the {\em new regressor} $\Phi_{21}(t) \in \rea$ is obtained from the solution of the ordinary differential equation
\begequ
\lab{regdyn}
\begmat{\dot\Phi_{11} \\ \dot\Phi_{21}} =
\begmat{0 & u_1\\ u_2 \Delta & u_3}
\begmat{\Phi_{11} \\ \Phi_{21}},
\endequ
with initial conditions
\begequ
\label{ic}
\begmat{\Phi_{11}(0) \\ \Phi_{21}(0)}=\begmat{1 \\ 0},
\endequ
and {\em arbitrary} bounded signals $u_i(t) \in \rea,$ $i=1,2,3$.
\end{proposition}

\begin{proof}\rm
Define the scalar dynamics
\begequ
\label{dotz}
	\dot z  = u_2y+ u_3 z,\;z(0)=0,
\endequ
and note that, since $\theta$ is constant, we can write
\begequ
\lab{dotthe}
	\dot\theta  = u_1 (z-z) .
\endequ
Combining \eqref{dotz} and \eqref{dotthe}, and using \eqref{lre}, we can write the ``virtual" LTV system
\begequ
\label{dotx}
\dot x = A(t) x + b(t)
\endequ
with
$$
\begin{aligned}
x & :=\begmat{\theta \\ z}\\
A(t) &:=\begmat{0 & u_1(t) \\ u_2(t) \Delta(t) & u_3(t)},\\
b(t)&:=\begmat{-u_1(t) x_2(t)\\ 0},
\end{aligned}
$$
and initial conditions
\begequ
\lab{icx}
x(0) =\begmat{\theta \\ 0}.
\endequ
Following the GPEBO approach define the dynamics
\begsubequ
\lab{dyn_gpebo}
\begin{equation}
\lab{dotxi}
    \quad \quad \quad ~ \dot \xi  = A(t) \xi + b(t),\; \xi(0)={\bf 0}_{2 \times 1}
\end{equation}
\vspace{-0.5cm}
\begin{equation}
    \lab{dotphi}
\dot \Phi  = A(t) \Phi,\; \Phi(0)=I_2.
\end{equation}
\endsubequ
Clearly, the state transition matrix $\Omega(t,s)$ of $A(t)$ is given by $\Omega(t,s) = \Phi(t)\Phi^{-1}(s)$.

The error signal
\begequ
\label{e}
e:=\xi - x
\endequ
satisfies $\dot e=A(t)e$. Consequently, from \eqref{e} and the properties of the system matrix $\Phi$ defined in \eqref{dotphi}, we get
$$
\begin{aligned}
x &=\xi - \Phi e(0)\\
&=\xi + \Phi \begmat{\theta \\ 0}
\\
&=\xi +\begmat{\Phi_{11} \\ \Phi_{21}}\theta,
\end{aligned}
$$
where, to get the second identity, we have taken into account \eqref{icx} and the initial conditions in \eqref{dotxi}, and introduced the notation
$$
\Phi=\begmat{\Phi_{11} & \Phi_{12} \\ \Phi_{21} &  \Phi_{22}}.
$$
Note now that
$$
\begmat{y\\ z}=\begmat{ \Delta &0 \\ 0 &1}x=\begmat{ \Delta &0 \\ 0 &1}\Big(\xi +\begmat{\Phi_{11} \\ \Phi_{21}}\theta \Big).
$$
The proof is then completed defining
\begequ
\lab{caly}
\caly =\begmat{ \caly_1\\  \caly_2}:=\begmat{y\\ z}-\begmat{ \Delta &0 \\ 0 &1}\xi,
\endequ
noting that
$$
\caly_2=z- \xi_2
$$
is a measurable signal, and computing the dynamics of the first column of $\Phi$.
\end{proof}

The procedure to generate the new LRE described in Proposition \ref{pro2} is summarized in the diagram in Fig. \ref{fig:block}.

\begin{figure}[h]
    \centering
    \includegraphics[width = 0.35\textwidth]{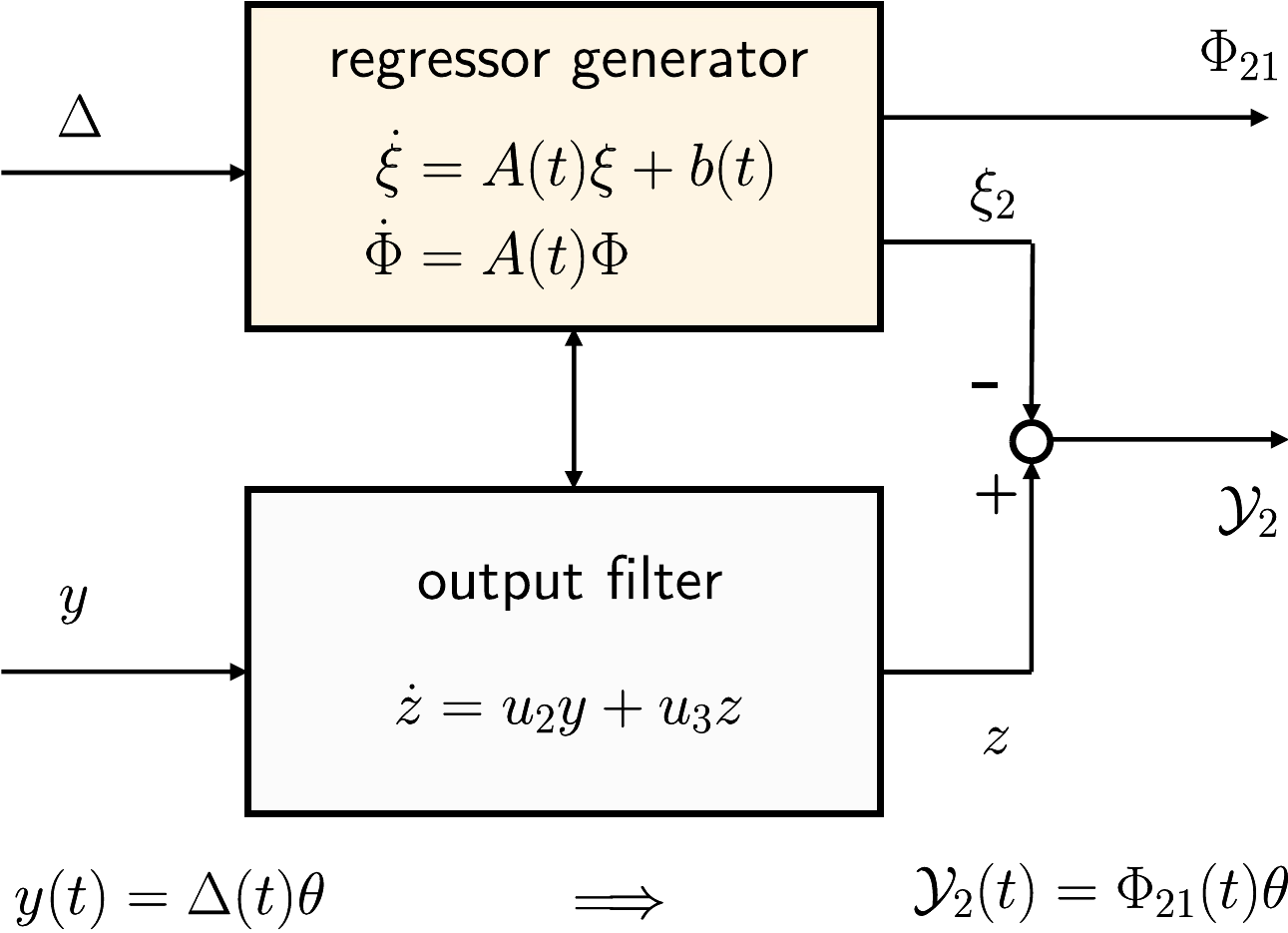}
    \caption{Block diagram of the proposed LRE generator}
    \label{fig:block}
\end{figure}

%
%%%%%%%%
\section{Improving the Excitation of the Regressor}
\label{sec4}
%%%%%%%%
%

\subsection{Main Result}

In this section we follow the basic idea of the energy pumping-and-damping injection construction used in \cite{YIetal} for orbital stabilization to select the signals $u_i$ in \eqref{regdyn} to provide excitation to the new regressor $\Phi_{21}$.

To articulate the main result we need the following assumption.

\begin{assumption}
\lab{ass1}
The {\em bounded} scalar signal $\Delta$ in \eqref{lre} is interval exciting (IE), \cite{KRERIE,TAObook}, that is there exist a time instant $t_c>0$ and a constant $\delta>0$ such that
$$
\int_{0}^{{t_c}}\Delta^2(\tau) d\tau \geq \delta.
$$
\end{assumption}

\begin{proposition}
\lab{pro3}
Assume $\Delta$ verifies Assumption \ref{ass1}. Define the signals
\begequ
\lab{u}
\begin{aligned}
u_1 & =- \alpha \Delta\\
u_2 & =\alpha \\
u_3 &= -\tilde  V(\Phi_{11},\Phi_{21}),
\end{aligned}
\endequ
with $\alpha(t) \in \rea$ a bounded signal such that $\alpha(t) \Delta(t) \not \equiv 0$ for some $t$,
\begequ
\lab{conalpdel}
\alpha \Delta \in \call_1,
\endequ
and
$$
\tilde V( \Phi_{11},\Phi_{21})  := \hal (\Phi_{11}^2 + \Phi_{21}^2) - \beta,
$$
where $ \Phi_{11},\Phi_{21}$ are the solutions of \eqref{regdyn} and $0 <\beta < \hal$. The resulting dynamics verifies the following properties.
\begenu
\item[{\bf F1}] Either $\Phi_{21} \not \in \call_2$ or
\begequ
\lab{bouphi}
 \Phi^2_{11}(t)+\Phi^2_{21}(t) \geq 2  \beta+\varepsilon, \quad \forall t \ge 0,
\endequ
for some (sufficiently small) $\varepsilon>0$.
\item[{\bf F2}]  The full state of the system---$\Phi,\xi$ and $z$---is bounded.
\endenu

\end{proposition}

\begin{proof}\rm
Replacing the signals $u_i$ in \eqref{u} to \eqref{regdyn} we get
$$
\begin{aligned}
\dot\Phi &={A(t)}\Phi\\
&=
\begmat{0 & -{\alpha} \Delta \\ {\alpha }\Delta &  -\tilde V}\Phi.
\end{aligned}
$$
The dynamics of the remaining states of the regressor generator, that is, $\xi$ and $z$, is given by
$$
\begmat{\dot\xi_1 \\ \dot \xi_2 \\ \dot z }=
\begmat{
0 & -\alpha\Delta & \alpha \Delta \\
 \alpha \Delta & - \tilde V & 0 \\
0 & 0& - \tilde V
}
\begmat{\xi_1 \\ \xi_2 \\ z}
+
\begmat{0 \\ 0 \\ \alpha y}.
$$
First, note that the IE Assumption \ref{ass1} rules out the extreme case $y(t)=\Delta(t)\equiv 0$. Indeed, in this case all signals of the regressor generator remain at their initial conditions, which is an equilibrium. To avoid this situation we also require the technical assumption that  $\alpha(t) \Delta(t) \not \equiv 0$ for some $t\ge 0$.

From the first column of the dynamic equation of $\Phi$ we immediately get
\begequ
\lab{dottilv}
\dot {\tilde V}=-\Phi^2_{21}\tilde V \leq 0,
\endequ
from which  we conclude boundedness of $(\Phi_{11}, \Phi_{21})$ as well as the invariance of the set
$$
\Omega:=\{q\in \rea^2|\tilde V(q) =0\}.
$$
Now, invoking the initial condition constraint \eqref{ic} yields
$$
\begin{aligned}
&  \quad \hal(\Phi^2_{11}(0)+\Phi^2_{21}(0))  =\hal \\
 \iff & \quad \tilde V( \Phi_{11}(0),\Phi_{21}(0))  + \beta = \hal \\
 \iff &\quad \tilde V( \Phi_{11}(0),\Phi_{21}(0))  = \hal  - \beta \\
  \implies & \quad  \tilde V( \Phi_{11}(0),\Phi_{21}(0))  > 0,
\end{aligned}
$$
where we have used the fact that $\beta < \hal$ to get the last implication. The latter inequality implies that the trajectory starts {\em outside} the disk described by the set $\Omega$. This, together with the invariance of the set implies that the {\em whole trajectory} $(\Phi_{11}(t), \Phi_{21}(t))$ is outside this disk, that is,
\begequ
\lab{outdis}
 \tilde V(t):=  \tilde V( \Phi_{11}(t),\Phi_{21}(t)) \geq 0,\;\forall t \geq 0.
\endequ

Let us, now, consider the two possible, mutually exclusive, cases:
\begenu
\item[{\bf (i)}] $\liminf \tilde V( t) = 0$
\item[{\bf (ii)}]   $\liminf \tilde V( t) \neq 0$.
\endenu

Consider, first, the second case and note that, due to the invariance of the set $\Omega$, this case also rules out the possibility that the convergence to zero happens in finite time. In this case the inequality \eqref{bouphi} holds true.

In the first case, integrating \eqref{dottilv} we see that the following equivalence holds true:
$$
\liminf \tilde V(t)=0\quad \iff \quad \Phi_{21} \not \in \call_2,
$$
completing the proof of the claim  {\bf F1}.

Now we show boundedness of the remaining states. For the second column of the matrix $\Phi$, note that it satisfies the dynamics
$$
\begin{aligned}
\begmat{\dot\Phi_{12} \\ \dot\Phi_{22}}&=
\begmat{0 & -{\alpha} \Delta \\ {\alpha }\Delta &  -\tilde V}
\begmat{ \Phi_{12} \\  \Phi_{22}},\; \begmat{\Phi_{12}(0) \\ \Phi_{22}(0)}=\begmat{0 \\ 1}.
\end{aligned}
$$
Consider the function
$$
W( \Phi_{12},\Phi_{22})  := \hal (\Phi_{12}^2 + \Phi_{22}^2) ,
$$
the derivative of which along the trajectories of the system \eqref{regdyn} satisfies
$$
\dot {W}=-\Phi^2_{22}\tilde V \leq 0,
$$
where the bound---from which we conclude the boundedness of $(\Phi_{12}, \Phi_{22})$---follows from \eqref{outdis}.

Finally, from \eqref{e} we have
$$
\begin{aligned}
\xi & = x + e\\
&= \begmat{\theta \\  z } +\Phi \begmat{\theta \\ 0 }.
\end{aligned}
$$
Hence, if $z$ is bounded we conclude that $\xi$ is also bounded completing the proof.

To prove boundedness of $z$ we, again, consider the two scenarios  {\bf (i)} and {\bf (ii)} described above. For the second case we consider the function
$$
V_z(z) := \hal z^2.
$$
Its time derivative along the trajectories of \eqref{regdyn} satisfies
$$
\begin{aligned}
	\dot V_z & = -\tilde V z^2 + \alpha y z\\
& \le - (\tilde V - {\epsilon \over 2})z^2 + {1 \over 2 \epsilon} (\alpha y)^2\\
& \le - (\varepsilon -\epsilon)V_z + {1 \over 2 \epsilon} (\alpha y)^2\\
& = - {\varepsilon \over 2}V_z + { 1 \over 4\varepsilon} (\alpha y)^2,
\end{aligned}
$$
where we have used the fact that, for all $\epsilon >0$, we have that
$$
\alpha y z \leq  {1 \over 2\epsilon}(\alpha y)^2+{\epsilon \over 2}z^2,
$$
to get the first inequality, the fact that $ \tilde V \geq {\varepsilon \over 2}$, which follows from \eqref{bouphi},  to get the second bound and selected $\epsilon={\varepsilon \over 2}$ in the final identity. The proof of boundedness of $z$ follows, then, from the inequality above and the fact that $\alpha y$ is bounded.

Now, consider the scenario {\bf (i)}. It is clear that the solution of the equation
\begequ
\lab{zequ}
\dot z = - \tilde V z + \alpha \Delta \theta,
\endequ
is given by
\begequ
\lab{solz}
z(t) =\left( \int_0^t e^{-\int_\sigma^t \tilde V(s)ds} \alpha(\sigma) \Delta(\sigma) d\sigma \right)\theta.
\endequ
From \eqref{solz} we immediately obtain
$$
\begin{aligned}
|z(t)| & \leq \bigg| \int_0^t e^{-\int_\sigma^t \tilde V(s)ds} \alpha(\sigma) \Delta(\sigma) d\sigma \bigg| |\theta|\\
 & \leq \int_0^t \Big|  e^{-\int_\sigma^t \tilde V(s)ds} \Big|  |\alpha(\sigma) \Delta(\sigma)| d\sigma|\theta|\\
  & \leq \int_0^t |\alpha(\sigma) \Delta(\sigma)| d\sigma|\theta|,
\end{aligned}
$$
where, to obtain the last bound, we have used the inequality \eqref{outdis}, which implies that
\begequ
\label{exp<1}
\left|  \exp \left(-\int_\sigma^t \tilde V(s)ds \right) \right| \leq 1.
\endequ
Boundedness of $z$ is concluded, from the inequality above, if  $\alpha \Delta \in \call_1$.
\end{proof}
\vspace{0.5cm}

Clearly, invoking the equivalence in  the claim {\bf P1} of Proposition \ref{pro1}, we have that the condition  $\Phi_{21} \not \in \call_2$ of the claim {\bf F1} of Proposition \ref{pro2} ensures {\em global convergence} of a gradient estimator with the new LRE. On the other hand, the inequality \eqref{bouphi} in the claim {\bf F1} guarantees the following excitation properties for the new regressor $\Phi_{21}$. If we can ensure that $\Phi_{11} \not \to \sqrt{2 \beta}$ then $\Phi_{21} \not \to 0$ and, consequently, $\Phi_{21}$ is PE.\footnote{Recall that if a scalar signal is PE then it is not square integrable, but the converse is not true \cite{ARAetal}.} In this sense, the new regressor $\Phi_{21}$ is ``more exciting'' than the original IE regressor $\Delta$. Although we have not been able to prove this property, it  has systematically been true in all our simulations, some of which are given in Section \ref{sec5}. Besides, in the following corollary we identify a scenario which guarantees that the new regressor $\Phi_{21} \notin \call_2$.

\begin{corollary}
\label{cor:1}
Consider the filter design in Proposition \ref{pro3}. If $\lim_{t\to\infty} \alpha(t) \Delta(t) =0$, then $\Phi_{21} \notin \call_2$.\footnote{It equivalently means that the second case in \eqref{bouphi} does not occur.}
\end{corollary}
%%%
\begin{proof}
We prove this fact by contradiction assuming that $\Phi_{21} \in \call_2$. According to the analysis in Proposition \ref{pro3}, $\tilde{V}$ has a limit when $t\to +\infty$, which can be zero or a positive constant, invoking non-negativeness of $\tilde V$ from \eqref{outdis}. Clearly, for the case $\Phi_{21} \in \call_2$, we have
$$
\lim_{t\to\infty} \tilde{V}(t) = \varepsilon
$$
for some $\varepsilon>0$. From $\Phi_{21} \in \call_2$ and $\dot{\Phi}_{21} \in \call_\infty$, we have $\Phi_{21} \to 0$, and consequently
$$
\lim_{t\to\infty} \Phi_{11}(t) = \sqrt{2(\beta + \varepsilon)}.
$$

Rewrite now $(\Phi_{11}, \Phi_{21})$ in polar coordinates, that is
$$
\Phi_{11}= r\cos\rho, \quad \Phi_{21}= r\sin\rho,
$$
with
\begequ
\label{conv:theta}
\begin{aligned}
     \rho & := \arctan\left({\Phi_{21} \over \Phi_{11}}\right),
     \\
     r &:= \sqrt{\Phi_{11}^2 + \Phi_{21}^2}.
\end{aligned}
\endequ
The transformation is well-posed and bijective almost globally except of the origin, which the above analysis proves that is not possible. In terms of the convergence of $\Phi_{21}$ and $\Phi_{22}$, we have
\begequ
\label{lim_rho}
\lim_{t\to\infty} \rho(t) = \rho_\star ,\quad
\lim_{t\to\infty} r(t) = \sqrt{2(\beta + \varepsilon)},
\endequ
for some constant $\rho_\star \in [0,2\pi)$. On the other hand, the dynamics of $\rho$ is given by
$$
\begin{aligned}
     \dot \rho & = {1\over 1+ ({\Phi_{21} \over \Phi_{11}})^2} \cdot {\dot{\Phi}_{21} \Phi_{11} - \dot{\Phi}_{11}\Phi_{22} \over \Phi_{11}^2}
     \\
     & = \alpha \Delta - {1\over 2} +{\beta \over r^2}.
\end{aligned}
$$
Invoking the convergence condition $\rho \to \rho_\star$ as $t\to\infty$, the boundedness of $\dot\rho$ and the $C^1$-smoothness of the solution $\rho(t)$, we have
$$
\lim_{t\to\infty} \dot\rho(t) = \lim_{t\to\infty} \left[\alpha(t) \Delta(t) - {1\over 2} +{\beta \over r^2(t)} \right]=0.
$$
From the convergence condition $\alpha\Delta \to 0$ and $r \to \sqrt{2(\beta + \varepsilon)}$, we have $\varepsilon=0$, which contradicts the assumption $\varepsilon>0$. Therefore, $\Phi_{21} \notin \call_2$, completing the proof.
\end{proof}

Two interesting observations are, first, that the condition  $\alpha \Delta \in \call_1$ cannot be replaced by  $\alpha \Delta \in \call_2$. Indeed, a counterexample to this claim is given by the selection
$$
 \tilde V(t)=e^{-t},\;\alpha(t) \Delta(t)={1\over t+1},
$$
for which $\alpha \Delta \in \call_2$, but $\alpha \Delta \not \in \call_1$. In this case the system \eqref{zequ} has unbounded solutions. Second, the condition  $\alpha \Delta \in \call_1$ is sufficient, but not necessary with a counterexample to this claim given by the selection
$$
 \tilde V(t)=e^{-t},\;\alpha(t) \Delta(t)={\sin t\over t+1},
$$
that does not satisfy the $\call_1$ condition but with associated solutions \eqref{zequ} bounded. Note also that ${\sin t \over t+1} \in \call_2$!

\subsection{Robustness}
\label{sec42}

We now consider robustness of the proposed scheme.\footnote{Robustness of the filter in Proposition \ref{pro2} is studied in the BIBO sense from the perturbed term to the filtered output, which is a property of the filter regardless of collected data.} For simplicity we gather all perturbations in a term $\delta(t)$, {\em i.e.}, the original LRE becomes
$$
y_{\tt N}(t) = \Delta(t) \theta + \delta(t),
$$
in which $(\cdot)_{\tt N}$ is used to represent the perturbed signal of the signal $(\cdot)$. Assuming that the selection of $\alpha$ is independent of the regression output $y$, the filter dynamics \eqref{dotz} becomes
\begequ
\label{dotzn}
\dot z_{\tt N} = u_2 y_{\tt N} + u_3z_{\tt N} , ~ z_{\tt N}(0)=0.
\endequ
From linearity we have
$$
z_{\tt N} = z + \delta_z,
$$
with $z$ the solution of $\dot z = u_2 y + u_3 z$ from $z(0)=0$ and
\begequ
\label{polre}
\delta_z(t) = \int_0^t \exp\left( \int_s^t - \tilde V(\tau) d\tau\right) \alpha(s)\delta(s) ds.
\endequ
Due to the above parameterization of the perturbation, $\delta$ does not affect $A(t)$, equivalently, the solution of $\Phi(t)$. On the other hand, the $\xi$-dynamics is cascaded to the filter \eqref{dotz}, {\em i.e.}, for the perturbed case
$$
\dot \xi_{\tt N} = A(t) \xi_{\tt N} + \begmat{\alpha \Delta z_{\tt N} \\ 0}
= A(t) \xi_{\tt N} + \begmat{\alpha \Delta (z + \delta_z) \\ 0},
$$
with $\xi_{\tt N}(0)= \mathbf{0}_{2\times 1}$. From superposition we have $\xi_{\tt N} = \xi + \delta_\xi$ with
$$
\delta_\xi(t) := \int_0^t \Phi(t)\Phi^{-1}(s) \alpha(s)\Delta(s) \begmat{\delta_z(s) \\0} ds.
$$
In terms of $\caly_2 = z- \xi_2$, we have the perturbed new LRE as
\begequ
\label{pnlre}
\caly_2 = \Phi_{21}\theta + \delta_{\tt N},
\endequ
with $\delta_{\tt N} := \delta_z - \delta_{\xi,2}$. Some observations are in order.
\begin{itemize}
\item[-] The perturbation $\delta$ from the original LRE \eqref{polre} affects the new LRE \eqref{pnlre} in a different way.

\item[-] Consider the perturbation $\delta$ as a zero-average high-frequency measurement noise. Noting that $\tilde V(t) \ge 0$ for all $t\ge0$, roughly speaking, the filter \eqref{dotz} plays a role similarly to a low-pass filter (which may be analyzed via standard averaging analysis), making the new LRE robust to the high-frequency component in $\delta$.

\item[-] For biased but bounded perturbations, {\em e.g.}, environmental disturbances and unmodelled dynamics, we have
$$
\begin{aligned}
|\delta_z| & \le \int_0^t \left|\exp\left( \int_s^t - \tilde V(\tau) d\tau\right)\right| |\alpha(s)||\delta(s)| ds
\\
& \le \int_0^t|\alpha(s)|\|\delta\|_{\infty} ds,
\end{aligned}
$$
in which we have used \eqref{exp<1}. Clearly, selecting $\alpha \in \call_1$ can guarantee $\delta_z \in \call_\infty$, and a vanishing signal $\alpha$ may also yield $\delta_\xi \in \call_\infty$, thus providing a BIBO stability property.
\end{itemize}

%
%%%%%%%%%%%%
\section{Simulations}
 \lab{sec5}
 %%%%%%%%%%%%%%5
 %

In this section we provide comprehensive simulation results to verify our main claims. In all simulations the parameter for pumping-and-damping injection is selected as $\beta = 0.4$ and---following the construction of the filter proposed above---the initial conditions  are fixed at
$$
z(0) = 0, \; \Phi(0)= I_2, \; \xi(0)= 0 .
$$
We first consider a constant unknown parameter $\theta = -5$, and test the proposed regressor generator under different excitation conditions for the original regressor $\Delta$. Namely, they are chosen as follows:
$$
\begin{aligned}
    \Delta_1(t) & = \hal e^{-t}\\
    \Delta_2(t) & =
    \left\{
    \begin{aligned}
         \hal & &t\in[0,5] \mbox{s}
         \\
         0 & & t> 5\mbox{s}.
    \end{aligned}
    \right.
\end{aligned}
$$
Note that these signals are clearly not PE and they belong to $\call_2$, hence the gradient estimator \eqref{graest} does not ensure parameter convergence. On the other hand, they satisfy the extremely weak condition that they are IE. For these regressor we use $\alpha(t) \equiv 1$, which verifies $\alpha\Delta \in \call_1$, as well as the additional condition in Corollary \ref{cor:1}.

To estimate the parameter we use the standard gradient descent adaptation with the old and the new regressors, that is,
 $$
 \dot{\hat \theta}_{\tt old} =  \gamma \Delta (y - \Delta \hat\theta_{\tt old}),
 $$
 and
  $$
 \dot{\hat \theta}_{\tt new} =  \gamma \Phi_{21} (\mathcal{Y}_2 - \Phi_{21} \hat\theta_{\tt new}),
 $$
selecting, in both cases, $\gamma = 2$. The initial conditions are selected as
$
\hat{\theta}_{\tt new}(0)=0$ 
and
$ \hat\theta_{\tt old}(0)=0.
$
As we can see from Figs. \ref{fig:1}-\ref{fig:2}, the proposed regressor generator transforms the interval exciting regressors into PE regressors ensuring that the estimate $\hat\theta_{\tt new}$  exponentially converges to the true value. On the other hand, the estimates generated with the original regressors $\hat\theta_{\tt old}$ exhibit a steady state error. In these figures, we also plot the evolution of the Frobenius norm of $\Phi(t)$, that is $\|\Phi\|_F = \sqrt{{\trace}(\Phi^\top \Phi)}$. It is also observed that all the states are bounded.
\begin{figure}[!htb]
    \centering
   \subfloat[Original and new regressors]{
   \includegraphics[width=0.23\textwidth]{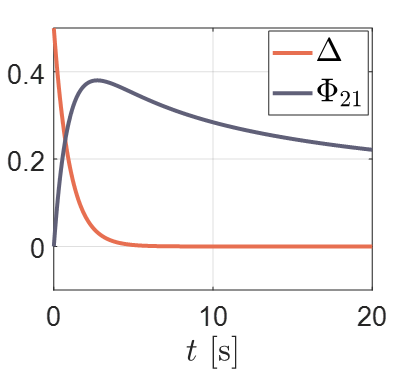}
   }
   \subfloat[Estimates using the new and original LRE]{
   \includegraphics[width=0.23\textwidth]{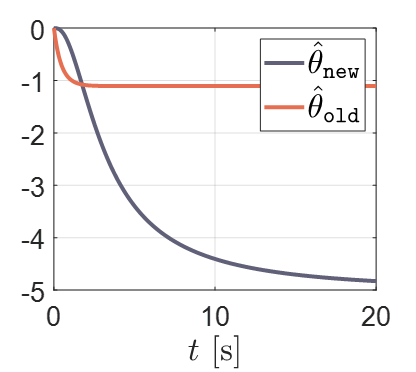}
   }

   \subfloat[Internal states]{
   \includegraphics[width=0.23\textwidth]{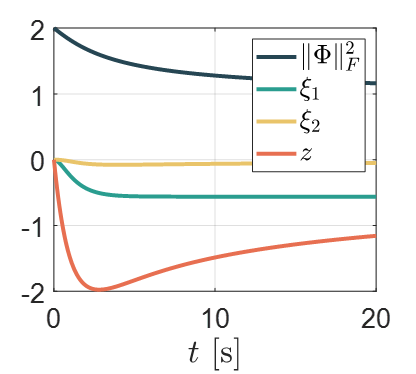}
   }
   \subfloat[The value of $\int_0^t|\alpha(s)\Delta(s)|$ds]{
   \includegraphics[width=0.23\textwidth]{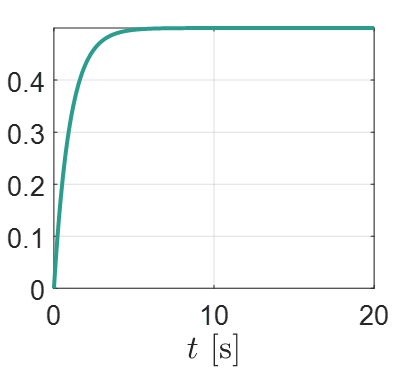}
   }
    \caption{Simulation results for the signal $\Delta_1(t) = \hal\exp(-t)$}
    \label{fig:1}
\end{figure}
%%%%%%%%%%%%%%%%%%%%%%
\begin{figure}[!htb]
    \centering
   \subfloat[Original and new regressors]{
   \includegraphics[width=0.23\textwidth]{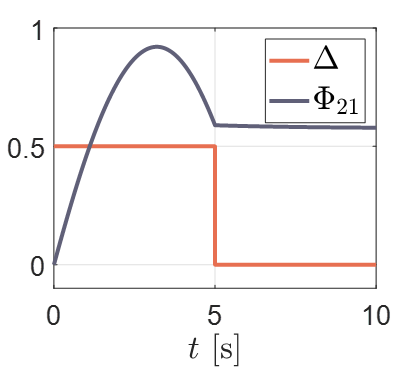}
   }
   \subfloat[Estimates using the new and original LREs]{
   \includegraphics[width=0.23\textwidth]{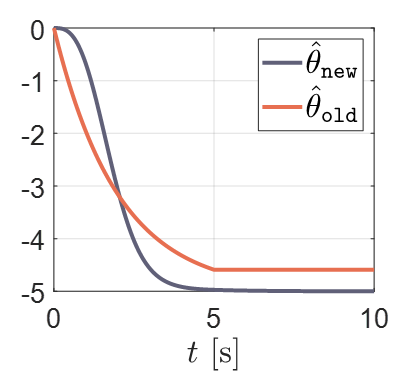}
   }

   \subfloat[Internal states]{
   \includegraphics[width=0.23\textwidth]{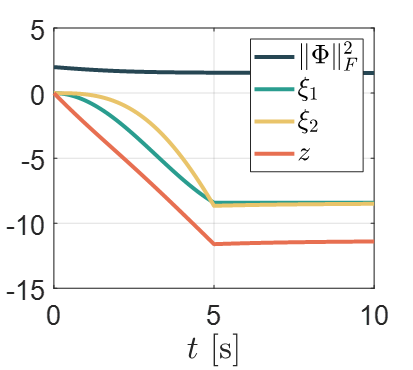}
   }
   \subfloat[The value of $\int_0^t|\alpha(s)\Delta(s)|$ds]{
   \includegraphics[width=0.23\textwidth]{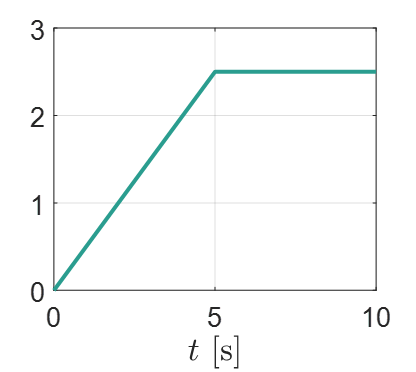}
   }
    \caption{Simulation results for the step signal $\Delta_2(t)$}
    \label{fig:2}
\end{figure}

Now, let us consider another vanishing signal
$$
   \Delta_3(t)  = {1\over 2t + 1} .
$$
It is clear that a non-zero constant $\alpha(t)$ cannot guarantee the condition $\alpha\Delta \in \call_1$. Therefore,  we parameterize $\alpha$ as
$$
\alpha(t) = \alpha_0(t) - k y(t)z(t),
$$
with $k>0$. Then, the key equation \eqref{dotz} takes the form
$$
\dot z = - (\tilde V + k y^2) z + \alpha_0 y,
$$
with the second term in $\alpha(t)$ providing additional damping in the dynamics of $z$. By carefully selecting $\alpha_0$, we may guarantee that $\alpha\Delta \in \call_1$. For $\Delta_3$ we select $k=0.1$ and $\alpha_0(t) = e^{- {t\over 10}}$ in simulation, see Fig. \ref{fig:3} for the corresponding results. It also illustrates the theoretical analysis.

We then add high-frequency noise in the LRE output $y$ via the ``Uniform Random Number'' block in Simulink/Matlab$^\text{TM}$. The signal-to-noise ratio (SNR) is selected around 20 dB with sampling time 0.01s. In this example we consider the signal $\Delta_2$. The simulation results are given in Fig. \ref{fig:5}, from which we observe that the new regressor output $\caly_2$ is hardly affected by the noise, as analyzed in Subsection \ref{sec42}.

\begin{figure}[!htb]
    \centering
   \subfloat[Regression outputs $y(t)$ and $\caly_2$]{
   \includegraphics[width=0.23\textwidth]{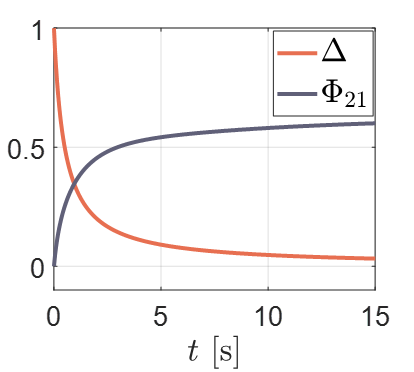}
   }
    \subfloat[Estimate using the new LRE]{
   \includegraphics[width=0.23\textwidth]{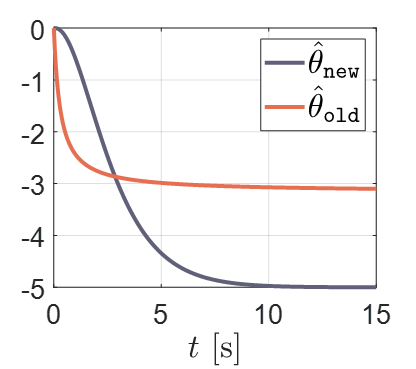}
   }
   \\
   \subfloat[Internal states]{
   \includegraphics[width=0.23\textwidth]{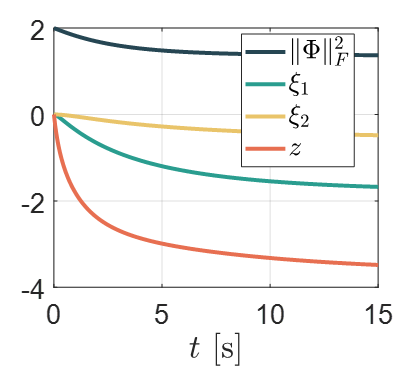}
   }
   \subfloat[The value of $\int_0^t|\alpha(s)\Delta(s)|$ds]{
   \includegraphics[width=0.23\textwidth]{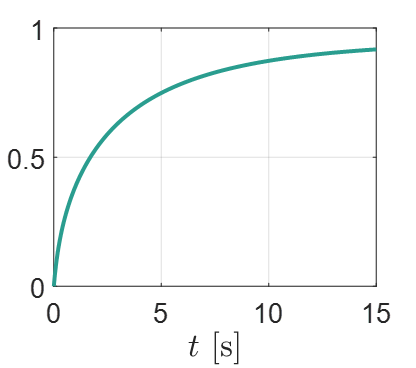}
   }
    \caption{Simulation results for the signal $\Delta_3(t)= {1\over 2t+1}$}
    \label{fig:3}
\end{figure}
%%%%%%%%%%%%%%%%%%%%%%
\begin{figure}[!htb]
    \centering
   \subfloat[Outputs in the original and new regressors]{
   \includegraphics[width=0.23\textwidth]{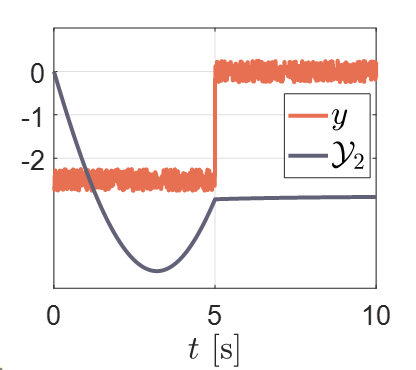}
   }
   \subfloat[Estimates using the new LRE]{
   \includegraphics[width=0.23\textwidth]{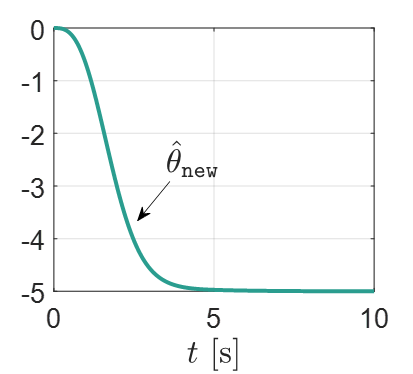}
   }
    \label{fig:5}
    \caption{Simulation results in the presence of high-frequency measurement noise for the signal $\Delta_2(t)$}
\end{figure}

As is well known, one of the main motivations of {\em on-line} identification is to track (slowly) time-varying parameters. To verify the alertness to the variation of the unknown parameter of the proposed procedure a simulation in which the parameter $\theta$ jumps from $-5$ to $-4$ at $t=10$s has been carried out with the IE regressor
$$   \Delta_4(t)  =
    \left\{
    \begin{aligned}
         \sin\left({\pi \over 10}t \right), & &t\in[0,12] \mbox{s}
         \\
         0, & & t> 12\mbox{s}
    \end{aligned}
    \right.
$$
and $\gamma =2$ and $\alpha(t)\equiv 1$.\footnote{We need to consider the signal being IE after the parameter variation for a while. Otherwise, the change of parameter cannot be captured by the output $y$.} The simulation results are shown in Fig. \ref{fig:5} from which we empirically conclude that the proposed new regressor is capable of dealing with parameter variations. The estimation $\hat{\theta}_{\tt old}$ from the original regressor is also presented, which, as we can see, fails to get satisfactory performance due to the lack of PE.

We underscore that there is an ultimate error from $\hat\theta_{\tt new}$, because Proposition \ref{pro2} only holds when $\theta$ is constant. For a time-varying $\theta$, a perturbation term appears in the new LRE \eqref{newlre}, inducing the observed estimation error. As future work, it would be interesting to study the ultimate error quantitatively. To overcome this issue, re-initialization of the filter in Proposition \ref{pro2} after parameter variation may be implemented, and it would be interesting to study how to re-initialize it automatically.
\begin{figure}[!htb]
    \centering
    \includegraphics[width=0.55\linewidth]{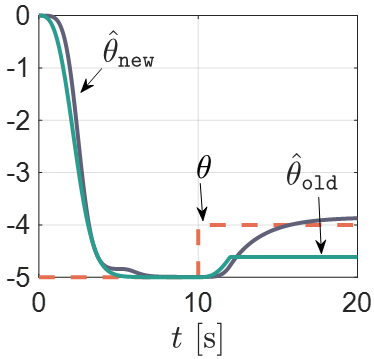}
    \caption{Alertness to time-varying parameter $\theta$ for the IE signal $\Delta_4(t)$}
    \label{fig:5}
\end{figure}

%
%%%%%%%%%%%%
\section{Conclusions and Future Research}
 \lab{sec6}
 %%%%%%%%%%%%%%5
 %
 We have presented a procedure to generate a new LRE in which the new regressor has some excitation properties, even if the original regressor does not. The key idea to carry out this task is the use of the virtual dynamics $\dot \theta=0$ given in \eqref{dotthe}. As seen in \eqref{regdyn}, the new LRE has some free signals $u_i$ that we selected in a particular way in Proposition \ref{pro3} to enforce the excitation-like condition of {\bf F1}. However, further research is needed to select other signals $u_i$ that would ensure {\em bona fide} excitation for a well-defined class of regressors $\Delta$---for instance, IE. It would be particularly interesting to relax the standing assumption \eqref{conalpdel}---that imposes constraints on the tuning function $\alpha$.

 A very simple extension of the result is obtained combining the elements of the signal $\caly$ in \eqref{caly} as
$$
Y:=\calh_1 (\caly_1)+\calh_2 (\caly_2),
$$
with $\calh_1$ and $\calh_2$ BIBO stable, scalar, LTV operators, to generate a new LRE
$$
Y=[\calh_1 (\Delta \Phi_{11}) + \calh_2( \Phi_{21})]\theta.
$$
A topic of future research is the selection of suitable operators $\calh_i$ to ensure excitation of the new regressor and boundedness of all signals. Note also that the extension of the results---including the application of the DREM procedure---to the case of nonlinear, {\em separable} parameterizations is straightforward.

Our LRE generator is restricted to the case of scalar LREs. However, mimicking the derivations of Section \ref{sec3} it is possible to generate a new LRE also for the vector case, that is, when the original LRE is of the form \eqref{veclre}. Indeed, using GPEBO it is possible to construct a measurable signal $e(t) \in \rea$ and a {\em new regressor} $\mu(t) \in \rea^q$, such that the new vector LRE
$$
e  = \mu^\top \theta
$$
holds. In this case  $\mu$ consists of the first $q$ elements of the last row of  the $(q+1) \times (q+1)$  matrix solution of the differential equation
$$
\dot\Phi=
\begmat{0_{q \times q} & u_1\\ u_2 \Delta^\top & u_3}\Phi,\;\Phi(0)=I_{q+1},
$$
with {\em arbitrary} functions $u_1(t) \in \rea^q$, $u_2(t) \in \rea$ and $u_3(t) \in \rea$. The problem is that selecting  the signals $u_i$ to stabilize this new LTV system seems a daunting task. In any case, given the availability of DREM, the interest of such a study is highly questionable.

Another interesting line of research that we are pursuing now is to apply, to the general problem of state observation, the idea of {\em virtual signal injection}---that combined with GPEBO---is exploited in this paper. Indeed, the key step in the construction of our LRE generator is captured in \eqref{dotthe}, that may be interpreted as a virtual signal injection. Some encouraging preliminary results have been developed and we expect to report them in the near future.

%
%%%%%%%%


\begin{thebibliography}{aaaaa}

\bibitem{araetalijacsp19}
S. Aranovskiy, A. Belov, R. Ortega, N. Barabanov and A, Bobtsov,  Parameter identification of linear time-invariant systems using dynamic regressor extension and mixing. {\it International Journal of Adaptive Control and Signal Processing}, vol. 33, no. 6, pp. 1016--1030, 2019.

\bibitem{ARAetal}
S. Aranovskiy, A. Bobtsov, R. Ortega and A. Pyrkin, Performance enhancement of parameter estimators via dynamic regressor extension and mixing,  \TAC, vol. 62, pp. 3546--3550, 2017. (See also {\tt arXiv:1509.02763} for an extended version.)

\bibitem{BARORT}
N. E. Barabanov and R. Ortega, On global asymptotic stability of $\dot x = \phi(t)\phi^\top (t)x$ with $\phi(t)$ bounded and not persistently exciting, {\it Systems \& Control Letters}, vol. 109, pp. 24--27, 2017.

\bibitem{CHOSHI}
N. Cho, H. Shin, Y. Kim and A. Tsourdos, Composite model reference adaptive control with parameter convergence under finite excitation, \TAC, vol. 63, pp. 811--818, 2018.

\bibitem{CHOetal}
G. Chowdhary, T. Yucelen, M. Muhlegg and E. N. Johnson, Concurrent learning adaptive control of linear systems with exponentially convergent bounds, {\em International Journal of Adaptive Control and Signal Processing}, vol. 27, no. 4, pp. 280--301, 2013.

\bibitem{CUIGAUANN}
Y. Cui, J. Gaudio and A. Annaswamy, A new algorithm for discrete-time parameter estimation, {\em ArXiv Preprint}, 2021. ({\tt arXiv:2103.16653})

\bibitem{IOASUNbook}
P. Ioannou and J. Sun, {\em Robust Adaptive Control}, Prentice-Hall, NJ, 1996.

\bibitem{KRERIE}
G. Kreisselmeier and G. Rietze-Augst, Richness and excitation on an interval---with application to continuous-time adaptive control, \TAC, vol. 35, no. 2, pp. 165--171, 1990.

\bibitem{LJUbook}
L. Ljung, {\em System Identification: Theory for the User}, Prentice Hall, New Jersey, 1987.

\bibitem{NARANNbook}
K. Narendra and A. Annaswamy, {\em Stable Adaptive Systems}, Prentice Hall, New Jersey, 1989.

\bibitem{ORTetalscl}
 R. Ortega, A. Bobtsov, A. Pyrkin and A. Aranovskyi: A parameter estimation approach to  state observation of nonlinear systems, {{\it Systems \& Control Letters}}, vol. 85, pp 84--94, 2015.

\bibitem{ORTetalaut}
R. Ortega, A. Bobtsov, N. Nikolaev, J. Schiffer and D. Dochain, Generalized parameter estimation-based observers: Application to power systems and chemical-biological reactors, \AUT, vol. 129, 109635, 2021.

\bibitem{ORTetaltac}
R. Ortega, S. Aranovskiy, A. Pyrkin, A Astolfi and A. Bobtsov, New results on parameter estimation via dynamic regressor extension and mixing: Continuous and discrete-time cases, \TAC, vol. 66, no. 5, pp. 2265--2272, 2021.

\bibitem{ORTNIKGER}
R. Ortega, V. Nikiforov and D. Gerasimov, On modified parameter estimators for identification and adaptive control: A unified framework and some new schemes, {\em IFAC Annual Reviews in Control}, vol. 50, pp. 278--293, 2020.

\bibitem{PANYU}
Y. Pan and H. Yu, Composite learning robot control with guaranteed parameter convergence, {\em Automatica}, vol. 89, pp. 398--406, 2018.

\bibitem{PRA}
L. Praly, Convergence of the gradient algorithm for linear regression models in the continuous and discrete-time cases, {\em Int. Rep. MINES ParisTech}, Centre Automatique et Syst\`emes, December 26, 2017.

\bibitem{ROYetal}
S.B. Roy, S. Bhasin and I.N. Kar, Combined MRAC for unknown MIMO LTI systems with parameter convergence, \TAC, vol. 63, pp. 283--290, 2018.

\bibitem{SASBODbook}
S. Sastry and M. Bodson, {\em Adaptive Control: Stability, Convergence and Robustness}, Prentice-Hall, New Jersey, 1989.

\bibitem{TAObook}
G. Tao, Adaptive control design and analysis. vol. 37. John Wiley \& Sons, New Jersey, 2003.

\bibitem{YIetal}
B. Yi, R. Ortega, D. Wu and W. Zhang, Orbital stabilization of nonlinear systems via Mexican sombrero energy pumping-and-damping injection, \AUT, vol. 112, 108861, 2020.

	
\end{thebibliography}
\end{document}